\documentclass[a4paper,11pt,reqno]{amsart}

\usepackage[dvipsnames,x11names]{xcolor}
\usepackage{hyperref}       
\usepackage{url}            
\usepackage{amsthm}
\usepackage{amsfonts}
\usepackage{amsmath}
\usepackage{amssymb}
\usepackage{color}
\usepackage{multirow}
\usepackage{float}
\usepackage[noadjust]{cite}
\usepackage{mathtools}
\usepackage{enumitem}
\usepackage{makecell}
\usepackage{array}
\newcolumntype{?}{!{\vrule width 1pt}}
\usepackage{tikz,pgfplots}
\pgfplotsset{compat=newest}
\usepackage{tikz}
\usetikzlibrary{arrows.meta}
\usetikzlibrary{positioning}
\usetikzlibrary{chains}
\usepackage{caption}

\date{}

\newcommand\giu[1]{{{\textcolor{red}{#1}}}}

\usepackage{bm,authblk}
\def\fa{\bm{f}}

\usepackage[margin=2.5cm]{geometry}

\theoremstyle{definition}

\newtheorem{theorem}{Theorem}[section]
\newtheorem{corollary}[theorem]{Corollary}
\newtheorem{lemma}[theorem]{Lemma}
\newtheorem{definition}[theorem]{Definition}
\newtheorem{proposition}[theorem]{Proposition}
\newtheorem{remark}[theorem]{Remark}
\newtheorem{notation}[theorem]{Notation}

\newtheorem{example}[theorem]{Example}

\def\mS{\mathcal{S}}

\def\mG{\mathcal{G}}
\def\mA{\mathcal{A}}

\def\N{\mathbb{N}}

\newcommand{\Bin}[2]{\binom{#1}{#2}_\R}

\newcommand{\Ppq}{\mathbb{P}_{p,q}}
\newcommand{\PP}{\mathbb{P}}
\newcommand{\PPn}{\mathbb{P}^n}

\newcommand{\K}{\mathcal{K}}
\def\F{\mathbb{F}_2^n}
\def\R{\mathbb{R}}
\def\<{\left<}
\def\>{\right>}
\def\l{\left(}
\def\r{\right)}
\def\FF{\mathbb{F}_2}
\def\wH{\omega^{\textnormal{H}}}
\def\dH{d^{\textnormal{H}}}

\def\HH{{\textnormal{H}}}
\def\WH{W^{\textnormal{H}}}

\def\ln{\textup{ln}}

\usepackage{amsaddr}

\usepackage{setspace}
\title{Parameters of Codes for the Binary Asymmetric Channel}

\author{Giuseppe Cotardo$^*$}
\thanks{$^*$ The author is supported by the Irish Research Council through grant n. GOIPG/2018/2534}

\address{School of Mathematics and Statistics \\ University College Dublin, Ireland}

\author{Alberto Ravagnani$^{**}$}
\thanks{$^{**}$ The author is supported by the Dutch Research Council through grants n. OCEANW.KLEIN.539 and n. VI.Vidi.203.045.}

\address{Department of Mathematics and Computer Science \\
Eindhoven University of Technology, 
the Netherlands}


\usepackage{verbatim}

\usepackage{cancel}

\begin{document}
\maketitle

\begin{abstract}
We introduce two notions of discrepancy between binary vectors, which are not metric functions in general but nonetheless capture the mathematical structure of 
the binary asymmetric channel. 
In turn, these lead to two new fundamental parameters 
of binary error-correcting codes, both of which 
measure the probability that the maximum likelihood decoder fails. We then derive 
various bounds for the cardinality and weight distribution of a binary code in terms of these new parameters, giving examples of codes meeting the bounds with equality.
\end{abstract}

\bigskip

\section*{Introduction}

In~\cite{curto2013combinatorial}, the authors initiate the study of coding theory in connection with neuroscience. They take a new look at neural coding from a mathematical perspective,
discretizing \textit{receptive fields} and modeling them as binary codes $C \subseteq \{0,1\}^n$. Under this approach, each of the $n$ coordinates of a vector represents a neuron, which is ``on'' if the corresponding coordinate is a~$1$, and ``off'' otherwise.

In the context of discretized receptive field codes, it is natural to assume that the  failure  of a neuron to be ``on'' is more probable than a neuron being ``on'' when it should not (we refer to \cite{curto2013combinatorial} for an explanation of the model). This behaviour is captured by the \textit{binary asymmetric channel}, in which $0$ becomes $1$ with probability~$p$ and~$1$ becomes~$0$ with probability $q \ge p$.
The binary asymmetric channel is in fact a family of channels, which comprises both the \textit{binary symmetric channel} (obtained for $p=q$) and the Z-channel (obtained for $p=0$).

In this work, we study the mathematical properties of the binary asymmetric channel and the structure of binary codes in connection with it. Very interesting progress has been recently made in this context, which we now briefly summarize.
In 2016 it was proved in~\cite{poplawski2016matched} that the binary asymmetric channel is metrizable in the weak sense of Massey~\cite{massey1967notes}. Shortly after that, Qureshi showed that the binary asymmetric channel admits a matched metric~\cite{qureshi2018matched}, establishing a conjecture of Firer and Walker~\cite{firer2015matched}. The approach of~\cite{qureshi2018matched} develops criteria  (and algorithms) to  understand  when  a  discrete  memoryless  channel  admits a  metric for which the  maximum  likelihood  decoder  coincides  with  the  nearest  neighbour decoder. A current open problem is to explicitly describe the metrics matching the binary asymmetric channel.

\bigskip

\subsection*{Our contribution.} \ In this paper, we bring forward the theory of  binary codes for the asymmetric channel, with a focus on their structural properties and parameters. We introduce two notions of 
\textit{discrepancy} between binary vectors $x,y \in \F$, which we denote by $\delta(x,y)$ and~$\hat\delta(x,y)$ respectively.
Although these are not metrics in general,
they are relatively simple functions that nicely relate to the probability law defining the binary asymmetric channel. In particular, as we will see,~$\delta$  matches such a channel (in the sense of our Theorem~\ref{th:decoders}). We also establish some general properties of~$\delta$ and~$\hat\delta$, showing the connection between these and the more traditional Hamming distance.
The function $\delta$ is not symmetric but satisfies the triangular inequality in general. In contrast,~$\hat\delta$ is symmetric but does not satisfy the triangular inequality.

Each of the two discrepancy functions mentioned above naturally defines a fundamental parameter of a code for the binary asymmetric channel. We call these \textit{minimum discrepancy} and \textit{minimum symmetric discrepancy}. Moreover, we show that they give incomparable upper bounds for the probability that the maximum likelihood decoder fails.

We devote the second part of the paper to 
bounds on the size of binary codes
having prescribed minimum (symmetric) discrepancy. We show how some first bounds can be  obtained from the theory of block codes endowed with the Hamming distance.
We then argue why techniques from classical coding theory do not easily extend to the discrepancy setting, showing that the graphs that naturally arise in the latter context do not have the same regularity as the corresponding ``Hamming-metric'' graphs. 
Finally, we derive upper and lower bounds for the size of binary codes using combinatorial arguments. The most interesting of these involves a code parameter that is finer than the cardinality, namely, the \textit{weight distribution}.  We illustrate how to apply the various bounds and provide examples of codes meeting them with equality.

\section{Preliminaries and Notation}
\label{sec:prel}

Throughout the paper, $\mathbb{F}_2=\{0,1\}$ is the binary field, $\N=\{0,1,2,\ldots\}$ is the set of natural numbers, and $n \ge 2$ is an integer. 
 We study a family of channels indexed by a pair of real numbers $(p,q)$ in the interval $[0,1/2)$.  These are defined as follows.

\begin{definition}
 Let $0\leq p\leq q< 1/2$ be real numbers and let
 \begin{equation*}
		\Ppq(1 \mid 0):=p, \qquad \Ppq(0 \mid 0) := 1-p, \qquad \Ppq(0 \mid 1):=q, \qquad \Ppq(1 \mid 1):=1-q.
	\end{equation*}
The \textbf{binary asymmetric channel} associated with $(n,p,q)$ is the triple $\K^n=(\F,\F,\Ppq^n)$, where $\Ppq^n:\F\times \F\longrightarrow\R$ is the function defined by 
\begin{equation}
    \label{eq:Prob}
    \Ppq^n(y\mid x)=\prod_{i=1}^n \Ppq(y_i\mid x_i) \quad \mbox{for all $x,y\in\F$}.
\end{equation}
\end{definition}

The previous definition models a 
discrete memoryless channel where the noise acts independently on the individual components of a binary vector. 
The assumption $p \le q$ tells us that
it is more probable that a $1$ becomes a $0$ than a $0$ becomes a $1$; see~Figure \ref{fig:BAC}. The function~$\Ppq^n$ expresses the transition probabilities when the channel is used $n$ times.

\begin{figure}[H]
    \centering
    \begin{tikzpicture}
    \tikzset{main node/.style={circle,fill=DodgerBlue1!20,draw,minimum size=0.8cm, inner sep=0pt}}

    	\node[main node] (1) {$0$};
    	\node[main node] (2) [right=4cm of 1] {$0$};
    	\node[main node] (3) [below=1.5cm of 1] {$1$};
    	\node[main node] (4) [right=4cm of 3] {$1$};
    
        \path[draw,thick,-{stealth},>=stealth,shorten >=1pt]
	    (1) edge node [above] {$1-p$} (2)
	    (1) edge node [above,pos=0.3] {$p$} (4)
	    (3) edge node [below,pos=0.3] {$q$} (2)
	    (3) edge node [below] {$1-q$} (4)
	    ;
    \end{tikzpicture}   
    \caption{}
    \label{fig:BAC}
\end{figure}

The main motivation for us 
to consider the binary asymmetric channel comes from recent developments in the theory of \textit{neural codes}.
More precisely, it has been shown in
\cite{curto2013combinatorial} that 
binary codes for the asymmetric channel
can be seen as a discretization of
\textit{receptive field codes}.
These are neural codes describing 
the brain's representation of the so called \textit{space of stimuli} covered by the \textit{receptive fields}.
In this context, binary codes are sometimes called \textit{combinatorial neural codes}.
We refer to \cite{curto2013neural,curto2013combinatorial} for further details and to~\cite[Figure~1]{curto2013neural} for a graphical representation of the link between stimuli and binary vectors.

Following~\cite{curto2013combinatorial}, we define codes for the binary asymmetric channel as follows.

\begin{definition}
    A 
    \textbf{code} is a subset $C \subseteq \F$ with $|C| \ge 2$. Its elements are called \textbf{codewords}.
\end{definition}

To simplify the discussion in the sequel, we introduce the following symbols.

\begin{notation}
We denote by $\wH(x) :=
|\{1 \le i \le n \, : \, x_i = 1\}|$ the \textbf{Hamming weight} of a binary vector  $x\in\F$. Moreover, for all $x,y \in \F$ we let 
\begin{align*}
    d_{00}(y,x)&:=|\{i\,:\,y_i=x_i=0\}|, & d_{01}(y,x)&:=|\{i\,:\,y_i=0 \mbox{ and } x_i=1\}|,\\
    d_{11}(y,x)&:=|\{i\,:\,y_i=x_i=1\}|, & d_{10}(y,x)&:=|\{i\,:\,y_i=1 \mbox{ and } x_i=0\}|.
\end{align*}
\end{notation}

We start with a preliminary result that gives a convenient expression for~$\Ppq^n$. We will need it later.

\begin{lemma}
\label{lem:Pxy}
	For all $x,y\in\F$ we have 
\begin{equation*}
		\Ppq^n(y \mid x)=\left(\frac{q}{1-p}\right)^{d_{01}(y,x)}\left(\frac{p}{1-q}\right)^{d_{10}(y,x)}(1-q)^{\wH(y)} \, (1-p)^{n-\wH(y)},
	\end{equation*}
	where we set $0^0:=1$.
\end{lemma}
\begin{proof}
Using the  definition of $\Ppq^n(y\mid x)$ directly, with the convention that $0^0=1$, one finds
\begin{equation}\label{eq:int}
    \Ppq^n(y\mid x)=(1-p)^{d_{00}(y,x)} \, p^{d_{10}(y,x)}\, q^{d_{01}(y,x)}\, (1-q)^{d_{11}(y,x)}.
\end{equation}
By definition we have
	\begin{equation*}
		d_{11}(y,x)=\wH(y)-d_{10}(y,x), \qquad d_{00}(y,x)=n-\wH(y)-d_{01}(y,x).
	\end{equation*}
In particular, the result immediately follows from \eqref{eq:int} when $p=0$.
If $p>0$ we write
	\begin{equation*}
		\begin{aligned}
			\Ppq^n(y\mid x)
			&=(1-p)^{n-\wH(y)-d_{01}(y,x)} \, p^{d_{10}(y,x)} \, q^{d_{01}(y,x)}\, (1-q)^{\wH(y)-d_{10}(y,x)}\\	&=\left(\frac{q}{1-p}\right)^{d_{01}(y,x)}\left(\frac{p}{1-q}\right)^{d_{10}(y,x)}(1-q)^{\wH(y)} \, (1-p)^{n-\wH(y)},
		\end{aligned}
	\end{equation*}
which is the desired expression.
\end{proof}

\begin{remark}
The binary asymmetric channel generalizes both the Z-channel and the binary symmetric channel, which are obtained for particular choices of the parameters~$p$ and~$q$. More precisely, $p=0$ (and $q$ arbitrary) gives the Z-channel, while $p=q$ gives the binary symmetric channel; see e.g.~\cite{constantin1979theory} and~\cite{cover2006theory}, respectively. \end{remark}

\bigskip
\section{Discrepancy}
\label{sec:dis}

In this section we introduce two functions that measure 
how ``different'' binary vectors are with respect to the binary asymmetric channel. We call these
\textit{discrepancy} and \textit{symmetric discrepancy}. We then prove that the 
maximum likelihood decoder for the binary asymmetric channel coincides 
with the decoder naturally associated with the first discrepancy function.

In the second part of the section we define two new parameters of a code $C \subseteq \F$ (one for each discrepancy notion) and establish their main properties. In later sections we will give evidence that these parameters measure the quality of a code for the binary asymmetric channel.

\begin{notation} 
\label{mainnot}
In the remainder of the paper we work with fixed real numbers $p$ and $q$ that satisfy 
$0 \le p \le q < 1/2$. We also let
\begin{equation*}
	\gamma:= \log_{\frac{q}{1-p}} \left(\frac{p}{1-q}\right) \in \R \cup \{+\infty\},
\end{equation*}
with the convention that 
$\gamma=+\infty$ if $p=0$ (and for any value of $q$).
\end{notation}

The quantity $\gamma$ (or better its inverse) has already been studied in connection with the binary asymmetric channel.
More precisely, the authors of \cite{qureshi2018equivalence} use the expression
\begin{equation*}
        S(p,q)=\frac{\ln(1-p)-\ln(q)}{\ln(1-q)-\ln(p)}=1/\gamma
    \end{equation*} 
    to classify
binary asymmetric channels up to equivalence. In this paper we will instead use~$\gamma$~to define a discrepancy function between binary vectors and new parameters of codes for the binary asymmetric channel. We start with the following numerical facts.

\begin{lemma} \label{lem:ineq}
	\begin{enumerate}[label={(\arabic*)}]
		\item We have $0\leq\frac{p}{1-q}\leq\frac{q}{1-p}< 1$. Moreover, $\frac{p}{1-q}=\frac{q}{1-p}$ if and only if $p=q$.
		\item We have $\gamma\ge 1$, with equality if and only if $0<p=q$.
	\end{enumerate}
\end{lemma}
\begin{proof}
	Since $0 \le p \le q<1/2$, we have
		$0 \le \frac{p}{1-q}< 1$ and $0 \le \frac{q}{1-p}< 1$. If $p=q$ then we clearly have $\frac{p}{1-q}=\frac{q}{1-p}$. On the other hand, if $p< q$ then $p+q<1$ implies $(p+q)(q-p) < q-p$. The latter inequality can be re-written as $p(1-p) < q(1-q)$, i.e., as $\frac{p}{1-q}<\frac{q}{1-p}$. Finally, since $\frac{p}{1-q}\leq\frac{q}{1-p}<1$, we have $\gamma\geq 1$.
\end{proof}

The first notion of discrepancy we propose is the following.

\begin{definition}
\label{def:dis}
	The \textbf{discrepancy} between  $y,x\in\F$ is
	$\delta_{p,q}(y,x):=\gamma_{p,q} \,  d_{10}(y,x)+d_{01}(y,x)$,
	with the convention that $+\infty \cdot 0=0$.
\end{definition}

Notice that if $0<p=q$ then $\gamma=1$ and therefore $\delta_{p,q}$ coincides with the \textbf{Hamming distance}~$\dH$ on~$\F$. The latter 
is  defined by 
$\dH(x,y):=|\{i \, : \,  x_i \neq y_i\}|$ for all $x,y \in \F$, which is clearly symmetric.
If $p<q$, then $\delta_{p,q}$ is
not a symmetric function in general. For example, if $p\neq 0$, $x=(1,1,1)$, and $y=(1,0,0)$, then $\delta_{p,q}(y,x)=2<2\gamma=\delta_{p,q}(x,y)$.

A natural way to decode a received message $y$ is to
return the codeword $x \in C$ that maximizes 
$\PP_{p,q}^n(y \mid x)$. The following definition is therefore standard in information theory\footnote{As in this paper we do not focus on complexity theory, decoders are defined as (deterministic) functions rather than algorithms.}.

\begin{definition} \label{def:mld}
For a code $C \subseteq \F$, the \textbf{maximum likelihood decoder} is the function $D_C:\F \to C \cup \{\fa\}$ defined by
$$D_C(y):= \left\{ 
\begin{array}{cl}
x & \mbox{if $x$ is the unique codeword that maximizes $\PP_{p,q}^n(y \mid x)$,} \\
\fa & \mbox{otherwise,}
\end{array}  \right.$$
where $\fa \notin \F$ denotes a failure message.
\end{definition}

The following result shows that the discrepancy function  $\delta_{p,q}$ \textit{matches} (in the sense of the natural generalization of a concept of S\'eguin to functions that are not necessarily distances; see~\cite{seguin1980metrics}) the  binary asymmetric channel. 

\begin{theorem} 
\label{th:decoders}
Let $x,x',y \in \F$. The following are equivalent:
\begin{enumerate}[label={(\arabic*)}]
\item \label{e1} $\delta_{p,q}(y,x) < \delta_{p,q}(y,x')$,
\item \label{e2} $\PPn_{p,q}(y \mid x) > \PPn_{p,q}(y \mid x')$.
\end{enumerate}
\end{theorem}
\begin{proof}
 By Lemma \ref{lem:Pxy}, proving that $\PPn(y \mid x')<\PPn(y \mid x)$ is equivalent to showing that
	\begin{equation} 
	\label{crit}
		\left(\frac{q}{1-p}\right)^{d_{01}(y,x')}\left(\frac{p}{1-q}\right)^{d_{10}(y,x')}<\left(\frac{q}{1-p}\right)^{d_{01}(y,x)}\left(\frac{p}{1-q}\right)^{d_{10}(y,x)},
	\end{equation}
	where $0^0=1$ by convention.
	Assume $p>0$. Then using the definition of discrepancy the inequality in \eqref{crit} can be re-written as
\begin{equation*}
	\left(\frac{q}{1-p}\right)^{\delta_{p,q}(y,x')} < \left(\frac{q}{1-p}\right)^{\delta_{p,q}(y,x)}.
\end{equation*}
Since $\frac{q}{1-p}<1$ by Lemma~\ref{lem:ineq}, the latter inequality holds if and only if $\delta_{p,q}(y,x)<\delta_{p,q}(y,x')$, as desired.

Now assume $p=0$ and $q$ arbitrary (possibly $q=0$ as well). 
Then \eqref{crit} is equivalent to
$$q^{d_{01}(y,x')} \, 0^{d_{10}(y,x')} < 
q^{d_{01}(y,x)} \, 0^{d_{10}(y,x)}.$$
This happens if and only if one of the following holds:
\begin{itemize}
    \item $d_{10}(y,x')=d_{10}(y,x)=0$ and 
    $d_{01}(y,x')>d_{01}(y,x)$,
    \item $d_{10}(y,x)=0$ and $d_{10}(y,x')>0$.
\end{itemize}
By definition of $\delta_{p,q}$, this is equivalent to
$\delta_{p,q}(y,x) < \delta_{p,q}(y,x')$,
concluding the proof.
\end{proof}

Theorem \ref{th:decoders} shows that, for any code $C \subseteq \F$,
the maximum likelihood decoder $D_C$ coincides with the 
\textbf{minimum discrepancy decoder} $D^\delta_C:\F \to C \cup \{\fa\}$. This is defined by
 $$D^\delta_C(y):= \left\{ 
\begin{array}{cl}
x & \mbox{if $x$ is the unique codeword that minimizes $\delta_{p,q}(y,x)$,} \\
\fa & \mbox{otherwise,}
\end{array}  \right.$$
where $\bm{f}$ is the same failure message as in Definition~\ref{def:mld}.

\begin{remark}
It is natural to compare the minimum discrepancy decoder $D^\delta_C$ (or equivalently the maximum likelihood decoder $D_C$) with the minimum Hamming distance decoder, denoted by $D_C^\HH$. These two decoders are different in general.
For example, let $p=0.1$ and $q=0.4$. Then $\gamma_{p,q}\approx 2.21$. Let 
	\begin{equation*}
		\begin{aligned}
			C=\{(0,0,0),(0,1,0),(1,1,0),(1,1,1)\} \subseteq \FF^3
		\end{aligned}
	\end{equation*}
	and $y=(0,0,1)$. One can check that $D^\delta_C(y)=(1,1,1)$, while $D_C^\HH(y)=(0,0,0)$.
\end{remark}

\begin{notation}
In the reminder of the paper we focus on the structure of codes endowed with the discrepancy function $\delta_{p,q}$.
For this type of study we will need to exclude the extreme case of the Z-channel from our treatment. Therefore, from now on, we will always assume
\begin{center}
    \framebox{\,$0<p \le q < 1/2$.}
\end{center}
To simplify the notation, we will also omit the subscript ``$p,q$'' in symbols, writing for example~$\PPn$ for~$\PPn_{p,q}$ and $\delta$ for $\delta_{p,q}$.
\end{notation}

We continue by introducing a second discrepancy function, which we denote by $\hat\delta$. In Section~\ref{sec:PUD} we will use both notions of discrepancy ($\delta$ and $\hat\delta$) to estimate the failure probability of the maximum likelihood decoder.

\begin{definition}
\label{def:symdis}
     \textbf{symmetric discrepancy} between vectors  $y,x\in\F$ is
	$\hat\delta(y,x):=\delta(y, x)-\wH(y)(\gamma-1)$.
\end{definition}

As we will see throughout the paper, the  functions $\delta$ and $\hat\delta$ have very different mathematical properties. For example, while $\delta$ is not symmetric but satisfies the triangular inequality, $\hat\delta$ is symmetric but does not satisfy the triangular inequality.

We now show
that $\hat\delta$ is indeed a symmetric function. We start with the following preliminary result.

\begin{lemma}
	\label{lem:relations}
	Let $x,y\in\F$. The following hold:
	\begin{enumerate}[label={(\arabic*)}]
		\item \label{dd} $d_{10}(y,x)=d_{10}(x,y)+\wH(y)-\wH(x)$,
		\item \label{ee} $d_{01}(y,x)=d_{01}(x,y)+\wH(x)-\wH(y)$,
		\item \label{ff} $\delta(y,x)=\delta(x,y)+(\wH(y)-\wH(x))(\gamma-1)$.
		\item \label{gg} $\delta(y,x)=\dH(y,x)+(\gamma-1)d_{10}(y,x)$.
		\item \label{hh} $\hat\delta(y,x)=\dH(y,x)-(\gamma-1)d_{11}(y,x)$. 
	\end{enumerate}
\end{lemma}
\begin{proof}
In order to prove the first equality, observe that $d_{10}(y,x)+d_{11}(y,x)=\wH(y)$, from which 
	\begin{equation*}
		d_{10}(y,x) =\wH(y)-d_{11}(y,x)=\wH(y)-d_{11}(x,y)=\wH(y)-\wH(x)+d_{10}(x,y).
	\end{equation*}
	Analogously, the second equality follows from the fact that  $d_{01}(y,x)+d_{11}(y,x)=\wH(x)$. The third equality is a consequence of the first and the second. The fourth equality follows from the fact that $\dH(y,x)=d_{10}(y,x)+d_{01}(y,x)$. Finally, combining the fourth equality with $d_{10}(y,x)+d_{11}(y,x)=\wH(y)$, we have
	\begin{equation*}
	    \hat\delta(y,x)=\delta(y,x)-\wH(y)(\gamma-1)=\dH(y,x)-(\gamma-1)d_{11}(y,x).
	\end{equation*}
	This concludes the proof.
\end{proof}

The following result follows from 
Lemma~\ref{lem:relations}\label{hh} and the fact that $\dH(y,x)=\dH(x,y)$ and $d_{11}(x,y)=d_{11}(y,x)$ for all $x,y\in\F$

\begin{proposition}
\label{prop:deltahat}
	For all $x,y\in\F$ we have $\hat\delta(y, x)=\hat\delta(x, y)$.
\end{proposition}

In the remainder of the section we turn to the structure of codes  $C \subseteq \F$. Each discrepancy notion ($\delta$ and $\hat\delta$)
defines a code parameter as follows.

\begin{definition}
For a code $C$, let
\begin{align*}
		\delta(C) &:= \min\{\delta(x,x') \, : \,  x,x' \in C, \, x \neq x'\}, \\
		\hat\delta(C)&:= \min\{\hat\delta(x,x') \, : \, x,x' \in C, \, x \neq x'\}.
\end{align*}
We call these the \textbf{minimum discrepancy} and the 
\textbf{minimum symmetric discrepancy} of~$C$, respectively.
\end{definition}

\begin{remark}
\label{rem:d2leqd1}
	For any code $C$ we have $\hat\delta(C)\leq\delta(C)$. Indeed, if $x,x'\in C$ satisfy $\delta(x,x')=\delta(C)$, then $\hat\delta(C)\leq \hat\delta(x,x') =\delta(x,x')-\wH(x)(\gamma-1)\leq\delta(x,x')=\delta(C)$. The numbers $\delta(C)$ and $\hat\delta(C)$ are very different in general. Moreover, while $\delta(C)$ is always a non-negative number, $\hat\delta(C)$ can even be negative.
\end{remark}

\begin{example}
   Let $p:=0.1$ and $q:=0.3$, from which $\gamma\approx1.77$. One can check that for  $C=\{(1,0,0),(0,1,1),(1,1,1)\}$ we have $\delta(C)=1$ and $\hat\delta(C)=\hat\delta((0,1,1),(1,1,1))\approx-0.54$.
\end{example}

The following result gives sufficient conditions under which a vector $y \in \F$ decodes to a codeword $x \in C$, in terms of the minimum (symmetric) discrepancies of $C$. In Theorem~\ref{thm:PUD} we will use these conditions to obtain bounds for the probability that the maximum likelihood decoder fails.

\begin{proposition} 
\label{prop:suff}
Let $C \subseteq \F$ be a code. Let $x \in C$ and $y \in \F$. We have
$D_C(y)=x$, provided that one of the following holds:
\begin{enumerate}[label={(\arabic*)}]
\item \label{item1:suff} $\delta(y,x)  <  \frac{\delta(C)+(\wH(y)-\wH(x))(\gamma-1)}{2}$, or
\item \label{item2:suff} $\delta(y,x) < \frac{\hat\delta(C)+\wH(y)(\gamma-1)}{2}$ .
\end{enumerate}
\end{proposition}

The proof of Proposition \ref{prop:suff} relies on the following triangular inequality for $\delta$.

\begin{lemma}
\label{lem:triandelta}
For all $x,y,z \in \F$ we have
\begin{equation*}
	\delta(z,x)\le \delta(z,y) + \delta(y,x).
\end{equation*}
\end{lemma}

\begin{proof}
Since the discrepancy is additive on the vector components, it is enough to prove the result for $n=1$. The case-by-case analysis is summarized in the following table:
\begin{center}
\begin{tabular}{|c|c|c?c|c|c|}
\hline
$a$ & $b$ & $c$ & $\delta(c,a)$ & $\delta(c,b)$ & $\delta(b,a)$ \\
\Xhline{3\arrayrulewidth}
0 & 0 & 0 & 0  & 0  & 0   \\ 
\hline
1 & 0 & 0 & $1$  & 0  & $1$   \\ 
\hline
0 & 1 & 0 & 0  & $1$  & $\gamma$   \\ 
\hline
0 & 0 & 1 &  $\gamma$ & $\gamma$  & 0   \\ 
\hline
1 & 1 & 0 & $1$  & $1$  & 0   \\ 
\hline
1 & 0 & 1 & 0  & $\gamma$  & $1$   \\ 
\hline
0 & 1 & 1 &  $\gamma$ & 0  &  $\gamma$  \\ 
\hline
1 & 1 & 1 & 0  & 0  & 0   \\ 
\hline
\end{tabular}
\end{center}
This concludes the proof by additivity.
\end{proof}

\begin{remark}
	Although $\hat\delta$ is symmetric, it does not satisfy a natural triangular inequality.
	More precisely, in general we have that
	$\hat\delta(z,x) \not\le \hat\delta(z,y) + \hat\delta(y,x)$.
	For example, suppose that $\gamma>1$ and let $x=(0,0,0)$, $y=(1,0,0)$, and $z=(1,1,0)$. Then 	$\hat\delta(z,x)=2>3-\gamma=\hat\delta(z,y)+\hat\delta(y,x)$.
	The closest ``triangular-type inequality'' we could derive for $\hat\delta$ is
	\begin{equation*}
	\label{eq:trianghat}
		\hat\delta(z,x)\le \hat\delta(z,y) + \hat\delta(y,x)+\wH(y)(\gamma-1),
	\end{equation*}
	which holds for all $x,y,z\in\F$ and contains $\wH(y)(\gamma-1)$ as correction term.
\end{remark}

\begin{proof}[Proof of Proposition \ref{prop:suff}]
We start by observing that, for any $x,x',y\in\F$, Lemmas~\ref{lem:relations} and~\ref{lem:triandelta} combined imply 
	\begin{equation}
	\label{eq:trianeqswap}
		\delta(y,x')\geq \delta(x,x')-\delta(y,x)+(\wH(y)-\wH(x))(\gamma-1).
\end{equation}
We now fix $x\in C$, $y\in\F$ and prove the two statements separately.
\begin{enumerate}[label={(\arabic*)}] 
    \item For every $x'\in C$ with $x' \neq x$, the inequality in~\eqref{eq:trianeqswap} and our assumption on $\delta(y,x)$ give
 	\begin{equation*}
		\begin{aligned}
			\delta(y,x')&\geq \delta(x,x')-\delta(y,x)+(\wH(y)-\wH(x))(\gamma-1)\\
			&>\delta(C)-\frac{\delta(C)+(\wH(y)-\wH(x))(\gamma-1)}{2}+(\wH(y)-\wH(x))(\gamma-1)\\
			&=\frac{\delta(C)+(\wH(y)-\wH(x))(\gamma-1)}{2}\\
			&>\delta(y,x).
		\end{aligned}
	\end{equation*}
	The desired statement now follows from Theorem \ref{th:decoders} and the definition of maximum likelihood decoder.
	\item 	Analogously, for every $x'\in C$ with $x' \neq x$, the inequality in~\eqref{eq:trianeqswap} and the assumption on $\delta(y,x)$ imply
	\begin{equation*}
		\begin{aligned}
			\delta(y,x')&\geq \delta(x,x')-\delta(y,x)+(\wH(y)-\wH(x))(\gamma-1)\\
			&\geq\hat\delta(C)-\frac{\hat\delta(C)+\wH(y)(\gamma-1)}{2}+\wH(y)(\gamma-1)\\
			&=\frac{\hat\delta(C)+\wH(y)(\gamma-1)}{2}\\
			&>\delta(y,x).
		\end{aligned}
	\end{equation*}
	Again, the statement follows from Theorem \ref{th:decoders}. \qedhere
\end{enumerate}
\end{proof}

\begin{remark}
In analogy with the Hamming distance, one may ask if $\delta(y,x)<\delta(C)/2$ implies $D_C(y)=x$, where $D_C$ is the maximum likelihood decoder; see Definition~\ref{def:mld}. This is not true in general. Take e.g. the code $C=\{(1,0,0),(0,1,1)\}$, with $p=0.1$ and $q=0.4$. We have $\delta(C)\approx 4.21$. Consider the vector $y=(0,0,0)$.
Then the set of codewords at discrepancy strictly less then $\delta(C)/2$ from $y$ is $C$, while the only codeword that minimizes the discrepancy from~$y$ is~$(1,0,0)$. 
\end{remark}

\begin{remark}
It is interesting to observe that
the condition in
Proposition~\ref{prop:suff}\ref{item1:suff}
is equivalent to $\dH(y,x)<\delta(C)/(\gamma+1)$,
which shows a connection between $\delta$ and the Hamming distance.
\end{remark}

\bigskip
\section{Unsuccessful Decoding}
\label{sec:PUD}
In this section we establish two upper bounds for the probability that the maximum likelihood decoder is unsuccessful. The first bound uses the notion of discrepancy (Definition~\ref{def:dis}), while the second uses the symmetric discrepancy (Definition~\ref{def:symdis}). We show that the two bounds are in general not comparable with each other, giving evidence that both discrepancy notions are relevant and independent concepts.

\begin{notation}
We let $\mS:=\{a+\gamma\, b \mid a,b \in \N \}$ and, for $h \in \N$, $\mS(h):=\{s \in \mS \mid 0\le s < h\}$. Moreover, for $a,b\in\R$ we let
	\begin{equation*}
         \Bin{a}{b}:=
         \begin{cases}
         \binom{a}{b} & \textup{if}\; a,b\in \N,\\
         \phantom{0}0 & \textup{otherwise}.
         \end{cases}
     \end{equation*}
\end{notation}

We start with the following preliminary result.

\begin{lemma}
\label{lem:isom}
Let $x,x' \in \F$ be vectors of the same Hamming weight. Let $i \in \{1,\ldots,n\}$ and $s \in \mS$. We have
\begin{equation*}
	|\{y \in \F \mid \wH(y)=i, \, \delta(y,x)=s \}| \, = \,  
|\{y \in \F \mid \wH(y)=i, \, \delta(y,x')=s \}|.
\end{equation*}
Moreover, if $j$ denotes the Hamming weight of $x$, then
\begin{equation*}
    |\{y \in \F \mid \wH(y)=i, \, \delta(y,x)=s \}| = \Bin{j}{\frac{i\gamma-s+j}{\gamma+1}}\Bin{n-j}{\frac{s-j+i}{\gamma+1}}.
\end{equation*}
\end{lemma}

\begin{proof}
Let $x \in \F$ be a vector and let $j=\wH(x)$. 
It is not difficult to see that the number of vectors $y \in \F$ with 
$\wH(y)=i$ and $\delta(y,x)=s$ is
$$\sum_{\substack{0 \le a \le j \\ 0 \le b \le n-j \\ a+b=i \\ j-a+\gamma b =s}} \binom{j}{a}\binom{n-j}{b}.$$
This number does not depend on $x$ but only on $j$. Moreover, the values of $a$ and $b$ are fully determined by the final constraints in the sum, and so the formula reduces to the product of binomial coefficients given in the statement of the lemma. This concludes the proof.
\end{proof}

\begin{notation}\label{not_lamba}
In the sequel, for $i,j \in \N$ and $s \in \R$ we denote by $\lambda(i,j,s)$ the number of vectors $y \in \F$ of weight~$i$
and with $\delta(y,x)=s$, where $x \in \F$ is any vector of weight $j$. By Lemma \ref{lem:isom}, this quantity $\lambda(i,j,s)$ is well-defined and given by
\begin{equation*}
    \lambda(i,j,s)=\Bin{j}{\frac{i\gamma -s+j}{\gamma+1}}\Bin{n-j}{\frac{s-j+i}{\gamma+1}}.
\end{equation*}
\end{notation}

Let $C\subseteq \F$ be a code, 
$x\in C$.
The probability that the maximum likelihood decoder $D_C$ returns 
a vector $x'\neq x$ or $\fa$ is 
\begin{equation*}
    \sum_{\substack{
    y \in\F\\D_C(y)\neq x
    }}\PPn(y \mid x)\giu{,}
\end{equation*}
where we recall that $\PPn(y \mid x)$ is interpreted as the probability that 
$y\in\F$ is received, given that
$x\in C$ was sent.

Under the (standard) assumption that all the codewords are transmitted with the same probability, we define the following concept.

\begin{definition}
\label{def:PUD}
     The \textbf{probability of unsuccessful decoding} for a code $C \subseteq \F$ is the average
\begin{equation*}
    \textup{PUD}(C)=\frac{1}{|C|} \; \sum_{x\in C}\sum_{\substack{
    y \in\F\\D_C(y)\neq x
    }}\PPn(y \mid x).
\end{equation*}
\end{definition}

We will study the $\textnormal{PUD}$ associated with a code $C$ in connection with its weight distribution. The latter is defined as follows.

\begin{definition}
\label{def:weightdist}
For a code $C \subseteq \F$ and an integer $i \in \{0,\ldots,n\}$, we denote by $\WH_i(C)$ the number of codewords $x \in C$ with Hamming weight $\wH(x)=i$. The tuple $(\WH_0(C),\ldots,\WH_n(C))$ is the \textbf{weight distribution} of $C$.
\end{definition}

We are now ready to establish the main result of this section.

\begin{theorem}
\label{thm:PUD}
Let $C \subseteq \F$ be a code. We have: 
\begin{enumerate}[label={(\arabic*)}]
    \item $\displaystyle \textup{PUD}(C) \le 1-\frac{1}{|C|}
\sum_{j=0}^n \WH_j(C) \sum_{i=0}^n (1-q)^i (1-p)^{n-i}\sum_{s \in \mS\l\frac{\delta(C)+(\gamma-1)(i-j)}{2}\r} \left( \frac{q}{1-p}  \right)^{s} \lambda(i,j,s)$,
    \item $\displaystyle \textup{PUD}(C) \le 1-\frac{1}{|C|}\sum_{j=0}^n \WH_j(C) \sum_{i=0}^n (1-q)^i (1-p)^{n-i}\sum_{s \in \mS\l\frac{\hat\delta(C)+i(\gamma-1)}{2}\r} \left( \frac{q}{1-p}  \right)^{s} \lambda(i,j,s) $. 
\end{enumerate}
\end{theorem}

\begin{proof}
We start by showing 
    the first bound in the statement. 
    By Lemma \ref{lem:Pxy} and the definitions of $\gamma$ and $\delta$ (Notation~\ref{mainnot} and  Definition~\ref{def:dis} respectively), we have 
\begin{align}
\textup{PUD}(C) &=
		1-\frac{1}{|C|}\sum_{x\in C}\sum_{\substack{y\in\F\\D(y)=x}}\PP^n(y\mid x) \nonumber \\
		&=1-\frac{1}{|C|}\sum_{x\in C}\sum_{\substack{y\in\F\\D(y)=x}} \left(\frac{q}{1-p}\right)^{\delta(y,x)}(1-q)^{\wH(y)}(1-p)^{n-\wH(y)}. \label{ppp}
\end{align}
Proposition \ref{prop:suff} implies
\begin{equation}
\label{eq:incl}
	\{y\in\F\;|\;\delta(y,x)<a(y,x)\}\subseteq\{y\in\F\;|\;D(y)=x\} \quad \mbox{for all $x\in C$},
\end{equation}
where
\begin{equation*}
	a(y,x)=\frac{\delta(C)+(\wH(y)-\wH(x))(\gamma-1)}{2}.
\end{equation*}
Combining \eqref{ppp} with \eqref{eq:incl} we then obtain 
\begin{equation*}
\begin{aligned}
	  \textup{PUD}(C)&\leq 1-\frac{1}{|C|}\sum_{x\in C}\sum_{\substack{y\in\F\\\delta(y,x)<a(y,x)}} \left(\frac{q}{1-p}\right)^{\delta(y,x)}(1-q)^{\wH(y)}(1-p)^{n-\wH(y)}\\
	&=1-\frac{1}{|C|}\sum_{j=0}^n\sum_{\substack{x\in C\\ \wH(x)=j}}\sum_{i=0}^n\sum_{s\in \mS\l\frac{\delta(C)+(\gamma-1)(i-j)}{2}\r}\sum_{\substack{y\in\F\\\wH(y)=i\\\delta(y,x)=s}}\left(\frac{q}{1-p}\right)^{s}(1-q)^{i}(1-p)^{n-i}.
\end{aligned}
\end{equation*}
Finally, using Notation \ref{not_lamba} we conclude that
\begin{equation*}
	\begin{aligned}
		  \textup{PUD}(C)&\leq 1-\frac{1}{|C|}\sum_{j=0}^n\sum_{\substack{x\in C\\ \wH(x)=j}}\sum_{i=0}^n\sum_{s\in \mS\l\frac{\delta(C)+(\gamma-1)(i-j)}{2}\r}\left(\frac{q}{1-p}\right)^{s}(1-q)^{i}(1-p)^{n-i}\lambda(i,j,s)\\
		&=1-\frac{1}{|C|}\sum_{j=0}^n\WH_j(C)\sum_{i=0}^n\sum_{s\in \mS\l\frac{\delta(C)+(\gamma-1)(i-j)}{2}\r}\left(\frac{q}{1-p}\right)^{s}(1-q)^{i}(1-p)^{n-i}\lambda(i,j,s),
	\end{aligned}
\end{equation*}
which is precisely the first bound in the statement of the theorem.
The proof of the second bound is analogous, using the fact that Proposition~\ref{prop:suff} implies 
\begin{equation}
\label{eq:inclb}
	\{y\in\F\;|\;\delta(y,x)<b(y,x)\}\subseteq\{y\in\F\;|\;D(y)=x\} \quad \mbox{for all $x\in C$},
\end{equation}
where
\begin{equation*}
	b(y,x)=\frac{\hat\delta(C)+\wH(y)(\gamma-1)}{2}. \qedhere
\end{equation*}
\end{proof}

\begin{example}
\label{ex:PUD_4}
    Let $C=\{(0,1,0,1,1),(1,1,0,0,0),(1,0,1,1,1)\}$ and $p=0.1$. Figure \ref{fig:PUD4} shows how the values of the two bounds of Theorem \ref{thm:PUD} change as $q$ ranges between $0.1$ and $0.49$. Moreover, one can also observe the relation between our bound and the trivial bound on the PUD stated above.
    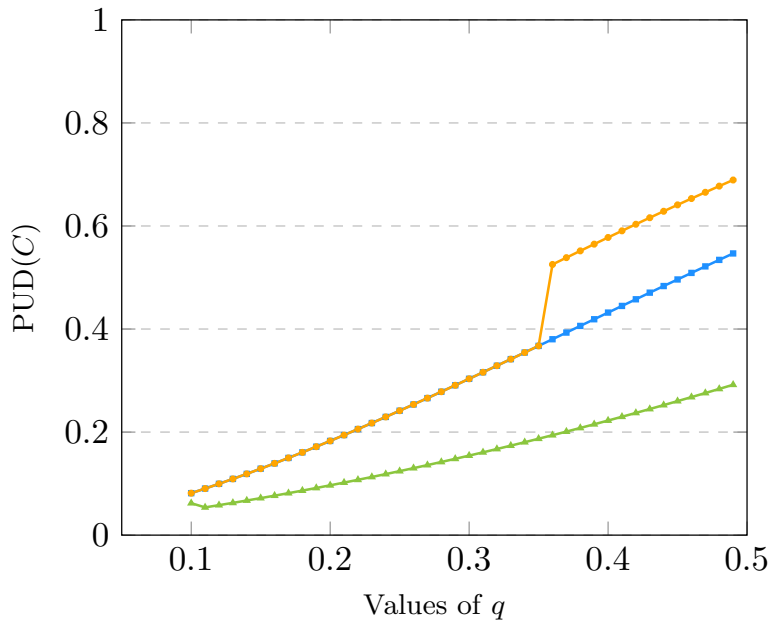
\begin{figure}[H]
\centering
\begin{tikzpicture}[scale=1.2]
\begin{axis}[
    xlabel={\footnotesize{Values of $q$}},
    ylabel={\footnotesize{PUD$(C)$}},
    xmin=0.05, xmax=0.5,
    ymin=0, ymax=1,
    xtick={0,0.1,0.2,0.3,0.4,0.5},
    ytick={0,0.2,0.4,0.6,0.8,1},
    legend pos=north west,
    ymajorgrids=true,
    grid style=dashed,
    every axis plot/.append style={thick},
    every mark/.append style={solid},
]
\addplot+[color=DodgerBlue1,mark=square*, mark size=0.6pt]
coordinates {
(0.100000, 0.0814600)
(0.110000, 0.0903676)
(0.120000, 0.0995849)
(0.130000, 0.109098 )
(0.140000, 0.118893 )
(0.150000, 0.128955 )
(0.160000, 0.139273 )
(0.170000, 0.149832 )
(0.180000, 0.160619 )
(0.190000, 0.171621 )
(0.200000, 0.182827 )
(0.210000, 0.194222 )
(0.220000, 0.205796 )
(0.230000, 0.217537 )
(0.240000, 0.229431 )
(0.250000, 0.241469 )
(0.260000, 0.253638 )
(0.270000, 0.265927 )
(0.280000, 0.278325 )
(0.290000, 0.290822 )
(0.300000, 0.303407 )
(0.310000, 0.316069 )
(0.320000, 0.328799 )
(0.330000, 0.341587 )
(0.340000, 0.354422 )
(0.350000, 0.367295 )
(0.360000, 0.380198 )
(0.370000, 0.393121 )
(0.380000, 0.406054 )
(0.390000, 0.418990 )
(0.400000, 0.431920 )
(0.410000, 0.444836 )
(0.420000, 0.457728 )
(0.430000, 0.470591 )
(0.440000, 0.483416 )
(0.450000, 0.496195 )
(0.460000, 0.508923 )
(0.470000, 0.521589 )
(0.480000, 0.534190 )
(0.490000, 0.546719 )
};

\addplot+[color=Orange1,mark=*,mark size=0.7pt]
coordinates {
(0.100000, 0.0814600)
(0.110000, 0.0903676)
(0.120000, 0.0995849)
(0.130000, 0.109098 )
(0.140000, 0.118893 )
(0.150000, 0.128955 )
(0.160000, 0.139273 )
(0.170000, 0.149832 )
(0.180000, 0.160619 )
(0.190000, 0.171621 )
(0.200000, 0.182827 )
(0.210000, 0.194222 )
(0.220000, 0.205796 )
(0.230000, 0.217537 )
(0.240000, 0.229431 )
(0.250000, 0.241469 )
(0.260000, 0.253638 )
(0.270000, 0.265927 )
(0.280000, 0.278325 )
(0.290000, 0.290822 )
(0.300000, 0.303407 )
(0.310000, 0.316069 )
(0.320000, 0.328799 )
(0.330000, 0.341587 )
(0.340000, 0.354422 )
(0.350000, 0.367295 )
(0.360000, 0.525350 )
(0.370000, 0.538556 )
(0.380000, 0.551692 )
(0.390000, 0.564750 )
(0.400000, 0.577720 )
(0.410000, 0.590595 )
(0.420000, 0.603367 )
(0.430000, 0.616027 )
(0.440000, 0.628568 )
(0.450000, 0.640983 )
(0.460000, 0.653264 )
(0.470000, 0.665405 )
(0.480000, 0.677399 )
(0.490000, 0.689238 )
};

\addplot+[color=LimeGreen,mark=triangle*,mark size=0.7pt]
coordinates {
(0.100000, 0.0620203)
(0.110000, 0.0538406)
(0.120000, 0.0581808)
(0.130000, 0.0626268)
(0.140000, 0.0671759)
(0.150000, 0.0718298)
(0.160000, 0.0765877)
(0.170000, 0.0814514)
(0.180000, 0.0864229)
(0.190000, 0.0914965)
(0.200000, 0.0966797)
(0.210000, 0.101970 )
(0.220000, 0.107365 )
(0.230000, 0.112865 )
(0.240000, 0.118473 )
(0.250000, 0.124188 )
(0.260000, 0.130009 )
(0.270000, 0.135936 )
(0.280000, 0.141967 )
(0.290000, 0.148105 )
(0.300000, 0.154346 )
(0.310000, 0.160693 )
(0.320000, 0.167145 )
(0.330000, 0.173698 )
(0.340000, 0.180354 )
(0.350000, 0.187112 )
(0.360000, 0.193973 )
(0.370000, 0.200931 )
(0.380000, 0.207990 )
(0.390000, 0.215146 )
(0.400000, 0.222400 )
(0.410000, 0.229749 )
(0.420000, 0.237195 )
(0.430000, 0.244733 )
(0.440000, 0.252362 )
(0.450000, 0.260085 )
(0.460000, 0.267895 )
(0.470000, 0.275794 )
(0.480000, 0.283778 )
(0.490000, 0.291847 )
};
\end{axis}
\end{tikzpicture}
\caption{\label{fig:PUD4} The first (blue) and second (orange) bound 
of Theorem~\ref{thm:PUD} for the code of Example~\ref{ex:PUD_4}, $p=0.1$ and some values of $q$. The plots show that the two bounds are not comparable in general. We also include the exact value of the PUD (green).}
\end{figure}
\end{example}

Notice that the two bounds of Theorem~\ref{thm:PUD} are not comparable in general, as the following example shows.  
\begin{example}
Take $p=0.1$ and $q=0.3$.
\begin{enumerate}[label={(\arabic*)}]
    \item Let $C=\{(0,0,1),(0,1,0),(1,1,1)\}$  and we have $\mbox{PUD}(C)\approx 0.34$. The first bound in Theorem \ref{thm:PUD} for   $\mbox{PUD}(C)$ is 
    $\approx 0.51$, while the second bound is $\approx 0.36$.
	\item Let $C=\{(0,0,0),(0,1,1),(1,1,1)\}$  and we have $\mbox{PUD}(C)\approx 0.24$. The first bound in Theorem~\ref{thm:PUD} for   $\mbox{PUD}(C)$ is $\approx 0.50$, while the second bound is $\approx 0.74$.
\end{enumerate}
\end{example}

It turns out that the two bounds of Theorem~\ref{thm:PUD} are incomparable even when $p$ and $C$ are fixed,
and $q$ varies. We illustrate this with the following examples.

\begin{example}
\label{ex:PUD}
    Let $C=\{(1,1,1,1),(1,0,0,1),(0,0,0,0)\}$ and $p=0.1$. Figure \ref{fig:PUD} shows how the values of the two bounds of Theorem~\ref{thm:PUD} change as $q$ ranges between $0.1$ and $0.49$.
    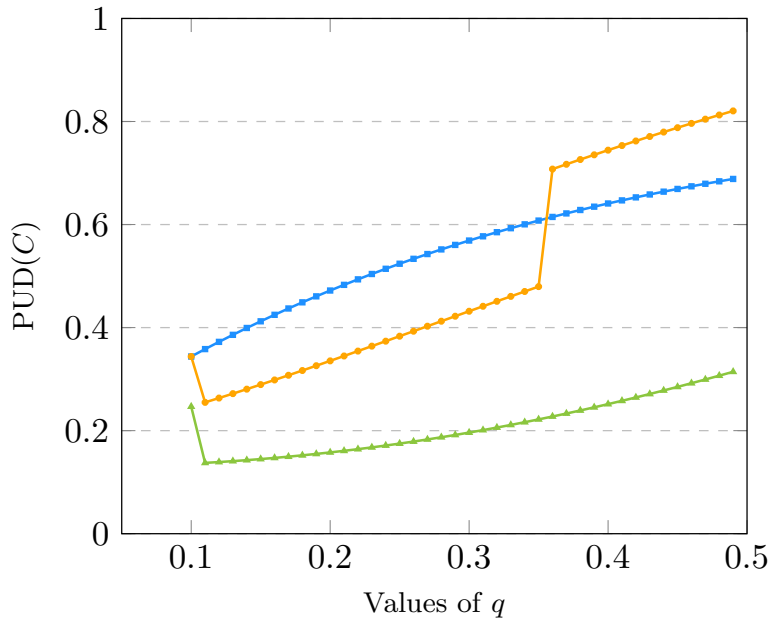
\begin{figure}[H]
\centering
\begin{tikzpicture}[scale=1.2]
\begin{axis}[
    xlabel={\footnotesize{Values of $q$}},
    ylabel={\footnotesize{PUD$(C)$}},
    xmin=0.05, xmax=0.5,
    ymin=0, ymax=1,
    xtick={0,0.1,0.2,0.3,0.4,0.5},
    ytick={0,0.2,0.4,0.6,0.8,1},
    legend pos=north west,
    ymajorgrids=true,
    grid style=dashed,
    every axis plot/.append style={thick},
    every mark/.append style={solid},
]
\addplot+[color=DodgerBlue1,mark=square*, mark size=0.6pt]
coordinates {
    ( 0.100000, 0.343900 )
    ( 0.110000, 0.358292 )
    ( 0.120000, 0.372313 )
    ( 0.130000, 0.385971 )
    ( 0.140000, 0.399272 )
    ( 0.150000, 0.412223 )
    ( 0.160000, 0.424831 )
    ( 0.170000, 0.437103 )
    ( 0.180000, 0.449045 )
    ( 0.190000, 0.460664 )
    ( 0.200000, 0.471967 )
    ( 0.210000, 0.482959 )
    ( 0.220000, 0.493649 )
    ( 0.230000, 0.504040 )
    ( 0.240000, 0.514141 )
    ( 0.250000, 0.523956 )
    ( 0.260000, 0.533493 )
    ( 0.270000, 0.542756 )
    ( 0.280000, 0.551752 )
    ( 0.290000, 0.560488 )
    ( 0.300000, 0.568967 )
    ( 0.310000, 0.577196 )
    ( 0.320000, 0.585180 )
    ( 0.330000, 0.592927 )
    ( 0.340000, 0.600439 )
    ( 0.350000, 0.607723 )
    ( 0.360000, 0.614784 )
    ( 0.370000, 0.621627 )
    ( 0.380000, 0.628258 )
    ( 0.390000, 0.634680 )
    ( 0.400000, 0.640900 )
    ( 0.410000, 0.646922 )
    ( 0.420000, 0.652750 )
    ( 0.430000, 0.658390 )
    ( 0.440000, 0.663846 )
    ( 0.450000, 0.669123 )
    ( 0.460000, 0.674225 )
    ( 0.470000, 0.679155 )
    ( 0.480000, 0.683920 )
    ( 0.490000, 0.688522 )
    };
\addplot+[color=Orange1,mark=*,mark size=0.7pt]
coordinates {
    ( 0.100000, 0.343900 )
    ( 0.110000, 0.254897 )
    ( 0.120000, 0.263278 )
    ( 0.130000, 0.271831 )
    ( 0.140000, 0.280541 )
    ( 0.150000, 0.289398 )
    ( 0.160000, 0.298388 )
    ( 0.170000, 0.307498 )
    ( 0.180000, 0.316716 )
    ( 0.190000, 0.326032 )
    ( 0.200000, 0.335433 )
    ( 0.210000, 0.344909 )
    ( 0.220000, 0.354446 )
    ( 0.230000, 0.364037 )
    ( 0.240000, 0.373668 )
    ( 0.250000, 0.383331 )
    ( 0.260000, 0.393015 )
    ( 0.270000, 0.402710 )
    ( 0.280000, 0.412406 )
    ( 0.290000, 0.422095 )
    ( 0.300000, 0.431767 )
    ( 0.310000, 0.441412 )
    ( 0.320000, 0.451023 )
    ( 0.330000, 0.460591 )
    ( 0.340000, 0.470108 )
    ( 0.350000, 0.479565 )
    ( 0.360000, 0.707655 )
    ( 0.370000, 0.716970 )
    ( 0.380000, 0.726205 )
    ( 0.390000, 0.735350 )
    ( 0.400000, 0.744400 )
    ( 0.410000, 0.753348 )
    ( 0.420000, 0.762188 )
    ( 0.430000, 0.770913 )
    ( 0.440000, 0.779518 )
    ( 0.450000, 0.787998 )
    ( 0.460000, 0.796347 )
    ( 0.470000, 0.804559 )
    ( 0.480000, 0.812631 )
    ( 0.490000, 0.820557 )
};
\addplot+[color=LimeGreen,mark=triangle*,mark size=0.7pt]
coordinates {
    ( 0.100000, 0.246700 )
    ( 0.110000, 0.137457 )
    ( 0.120000, 0.139068 )
    ( 0.130000, 0.140837 )
    ( 0.140000, 0.142765 )
    ( 0.150000, 0.144854 )
    ( 0.160000, 0.147106 )
    ( 0.170000, 0.149522 )
    ( 0.180000, 0.152103 )
    ( 0.190000, 0.154851 )
    ( 0.200000, 0.157767 )
    ( 0.210000, 0.160851 )
    ( 0.220000, 0.164105 )
    ( 0.230000, 0.167529 )
    ( 0.240000, 0.171123 )
    ( 0.250000, 0.174887 )
    ( 0.260000, 0.178823 )
    ( 0.270000, 0.182929 )
    ( 0.280000, 0.187205 )
    ( 0.290000, 0.191651 )
    ( 0.300000, 0.196267 )
    ( 0.310000, 0.201051 )
    ( 0.320000, 0.206003 )
    ( 0.330000, 0.211122 )
    ( 0.340000, 0.216406 )
    ( 0.350000, 0.221854 )
    ( 0.360000, 0.227465 )
    ( 0.370000, 0.233237 )
    ( 0.380000, 0.239168 )
    ( 0.390000, 0.245257 )
    ( 0.400000, 0.251500 )
    ( 0.410000, 0.257896 )
    ( 0.420000, 0.264442 )
    ( 0.430000, 0.271136 )
    ( 0.440000, 0.277974 )
    ( 0.450000, 0.284954 )
    ( 0.460000, 0.292073 )
    ( 0.470000, 0.299327 )
    ( 0.480000, 0.306713 )
    ( 0.490000, 0.314228 )
};
\end{axis}
\end{tikzpicture}
\caption{\label{fig:PUD} The first (blue) and second (orange) bound 
of Theorem~\ref{thm:PUD} for the code of Example~\ref{ex:PUD}, $p=0.1$ and some values of $q$. The plots show that the two bounds are not comparable in general. We also include the exact value of the PUD (green).}
\end{figure}
\end{example}

\bigskip
\section{Bounds}
\label{sec:bounds}

In this section we investigate the cardinality of a code $C \subseteq \F$ as a function of the fundamental parameters $n$, $\gamma$, and $\delta(C)$ or $\hat\delta(C)$. 
We obtain various bounds and compare them with each other. We also give examples of codes meeting them.

\begin{remark}
    Before presenting our statements and their proofs, we illustrate why several arguments classically applied to study codes with the Hamming metric do \textit{not} extend to the case of our interest (codes endowed with  discrepancy functions). Many of these arguments rely, implicitly or explicitly, on the regularity properties of certain graphs defined via the Hamming metric. Unfortunately, the natural analogues of these graphs for discrepancy functions do not exhibit the same regularity.
    
    Specifically, it is natural to interpret codes with minimum (symmetric) discrepancy bounded from below as cliques in the graphs $\mG(n,\delta)$ and 
    $\hat\mG(n,\delta)$ defined as follows. The vertices of both graphs are the elements of $\F$. Vertices $x$, $y$ are adjacent in $\mG(n,\delta)$
    if $\min\{\delta(y,x),\delta(x,y)\}\geq \delta$ {and $x\neq y$}, and are adjacent in $\hat\mG(n,\delta)$ if $\hat\delta(x,y)\geq \delta$ {and $x\neq y$}. When $\gamma=1$, the graph $\mG(n,\delta)$ is \textit{vertex-transitive},
    a property that greatly simplifies the study of code parameters via anticliques; see in particular~\cite{el2007bounds}.
    
   As the next two examples illustrate, the graphs $\mG(n,\delta)$ and $\hat\mG(n,\delta)$ are in general not vertex-transitive, preventing the 
    application of the arguments mentioned above.
\end{remark}

\begin{example} 
\begin{enumerate}[label={(\arabic*)}]
    \item Let $p=0.1$ and $q=0.2$, from which $\gamma\approx 1.38$. 
    A graphical representation of $\mG(4,\gamma+1)$, where the vertex $u$ corresponds to the unique element $x=(x_0,x_1,x_2,x_3)$ of $\mathbb{F}_2^4$ with $x_0+2x_1+4x_2+8x_3 = u$, can be found in Figure~\ref{f:1}. One can check that $\mG(4,\gamma+1)$ is not vertex-transitive (e.g., vertices $0$ and $1$ have different degrees).

\item 	Let $p=0.1$ and $q=0.3$, which gives $\gamma\approx 1.77$. A graphical representation of $\hat\mG(3,1)$ can be seen in Figure~\ref{f:2}, where the vertex $u$ corresponds to the unique element $x=(x_0,x_1,x_2)$ of $\mathbb{F}_2^3$ with $x_0+2x_1+4x_2= u$. Again,
$\hat\mG(3,1)$ is not vertex-transitive (again, vertices $0$ and $1$ have different degrees).
\end{enumerate}
\end{example}

\begin{figure}[h!]
\begin{center}
\begin{tikzpicture}
[scale=0.8,start chain=circle placed {at=(\tikzchaincount*22.5:6)},regular/.style={draw,circle,inner sep=0,minimum size=0.7cm,fill=DodgerBlue1!20},rotate=90]
	\foreach \i in {15,...,0}
 	 \node [on chain, regular] (\i) {\i};
 	 
 	\path[draw,thick]
 	(0) edge node {} (7)
 	(0) edge node {} (11)
 	(0) edge node {} (13)
 	(0) edge node {} (14)
 	(0) edge node {} (15)
 	
 	(1) edge node {} (2)
 	(1) edge node {} (4)
 	(1) edge node {} (6)
 	(1) edge node {} (8)
 	(1) edge node {} (10)
 	(1) edge node {} (12)
 	(1) edge node {} (14)
 	(1) edge node {} (15)
 	
 	(2) edge node {} (4)
 	(2) edge node {} (5)
 	(2) edge node {} (8)
 	(2) edge node {} (9)
 	(2) edge node {} (12)
 	(2) edge node {} (13)
 	(2) edge node {} (15)
 	
 	(3) edge node {} (4)
 	(3) edge node {} (5)
 	(3) edge node {} (6)
 	(3) edge node {} (8)
 	(3) edge node {} (9)
 	(3) edge node {} (10)
 	(3) edge node {} (12)
 	(3) edge node {} (13)
 	(3) edge node {} (14)
 	
 	(4) edge node {} (8)
 	(4) edge node {} (9)
 	(4) edge node {} (10)
 	(4) edge node {} (11)
 	(4) edge node {} (15)
 	
 	(5) edge node {} (6)
 	(5) edge node {} (8)
 	(5) edge node {} (9)
 	(5) edge node {} (10)
 	(5) edge node {} (11)
 	(5) edge node {} (12)
 	(5) edge node {} (14)
 	
 	(6) edge node {} (8)
 	(6) edge node {} (9)
 	(6) edge node {} (10)
 	(6) edge node {} (11)
 	(6) edge node {} (12)
 	(6) edge node {} (13)
 	
 	(7) edge node {} (8)
 	(7) edge node {} (9)
 	(7) edge node {} (10)
 	(7) edge node {} (11)
 	(7) edge node {} (12)
 	(7) edge node {} (13)
 	(7) edge node {} (14)
 	
 	(8) edge node {} (15)
 	
 	(9) edge node {} (10)
 	(9) edge node {} (12)
 	(9) edge node {} (14)
 	
 	(10) edge node {} (12)
 	(10) edge node {} (13)
 	
 	(11) edge node {} (12)
 	(11) edge node {} (13)
 	(11) edge node {} (14)
 	
 	(13) edge node {} (14)
 	;
\end{tikzpicture}
\end{center}
\caption{}
\label{f:1}
\end{figure}

\begin{figure}[h!]
\begin{center}
\begin{tikzpicture}
    [scale=0.8,start chain=circle placed {at=(\tikzchaincount*44.8:5)},regular/.style={draw,circle,inner sep=0,minimum size=0.7cm,fill=DodgerBlue1!20},rotate=90]
	\foreach \i in {7,...,0}
 	 \node [on chain, regular] (\i) {\i};
 	 
 	 \path[draw,thick]
	(0) edge node {} (1)
	(0) edge node {} (2)
	(0) edge node {} (3)
	(0) edge node {} (4)
	(0) edge node {} (5)
	(0) edge node {} (6)
	(0) edge node {} (7)
	
	(1) edge node {} (2)
	(1) edge node {} (4)
	(1) edge node {} (6)
	(1) edge node {} (7)
	
	(2) edge node {} (4)
	(2) edge node {} (5)
	(2) edge node {} (7)
	
	(3) edge node {} (4)
	(3) edge node {} (5)
	(3) edge node {} (6)
	
	(4) edge node {} (7)
	
	(5) edge node {} (6)
	;
\end{tikzpicture}
\end{center}
\caption{}
\label{f:2}
\end{figure}

First bounds for codes of  minimum discrepancy bounded from below
can be obtained from bounds for codes with the Hamming metric. We start by showing how
the functions $\delta$ and $\dH$ relate to each other.

\begin{proposition} \label{three}
Let $x,y \in \F$. The following hold:
\begin{enumerate}[label={(\arabic*)}]
    \item \label{itm:one} $\hat\delta(x,y) \le \dH(x,y) \le \delta(x,y)$,
    \item \label{itm:two} $\delta(x,y)+\delta(y,x) = (\gamma+1) \, \dH(x,y)$,
    \item \label{itm:three} $\min\{\delta(x,y), \delta(y,x)\} \le \frac{\gamma+1}{2} \, \dH(x,y) $.
\end{enumerate}
In particular, for all codes $C \subseteq \F$ we have
\begin{equation*}
    \hat\delta(C) \le \dH(C) \le \delta(C) \le \frac{\gamma+1}{2} \, \dH(C).
\end{equation*}
\end{proposition}

\begin{proof}
    Let $x,y\in\F$ and recall that $\gamma\geq 1$ by Lemma \ref{lem:ineq}. By Lemma \ref{lem:relations}\ref{hh}, we have 
    \begin{equation*}
         \hat\delta(x,y)=\dH(x,y)-(\gamma-1)d_{11}(x,y)\leq \dH(x,y)
    \end{equation*}
    and, by Lemma \ref{lem:relations}\ref{gg}, we have
    \begin{equation*}
        \dH(x,y)\leq \dH(x,y)+(\gamma-1)d_{10}(x,y)=\delta(x,y).
    \end{equation*}
   This implies \ref{itm:one}. Moreover, by the very definition of $\delta$ we have
	\begin{align*}
		\delta(x,y)+\delta(y,x)&= d_{01}(x,y)+\gamma\, d_{10}(x,y) +d_{01}(y,x)+\gamma\, d_{10}(y,x)\\
		&= d_{01}(x,y)+\gamma\, d_{10}(x,y) +d_{10}(x,y)+\gamma\, d_{01}(x,y)\\
		&=(\gamma+1)\,(d_{01}(x,y)+d_{10}(x,y))\\
		&=(\gamma+1)\,\dH(x,y),
	\end{align*}
	establishing \ref{itm:two}. We also have
	\begin{equation*}
	    2\min\left\{\delta(x,y),\delta(y,x)\right\} \leq \delta(x,y)+\delta(y,x)=(\gamma+1)\dH(x,y), 
	\end{equation*}
	which shows \ref{itm:three}. 
	To prove the last part of the statement, fix a code $C\subseteq\F$. Then \ref{itm:one} implies
	\begin{align*}
	    \dH(C) &= \min\{\dH(x,y)\mid x,y \in C\} \leq \min\{\delta(x,y)\mid x,y \in C\}=\delta(C),\\
	    \hat\delta(C)&=\min\{\hat\delta(x,y)\mid x,y \in C\} \leq \min\{\dH(x,y)\mid x,y \in C\}=\dH(C).
	\end{align*}
	Finally, let $x,y\in C$ such that $\dH(C)=\dH(x,y)$. By \ref{itm:three} we have
	\begin{equation*}
	    2\,\delta(C) \leq 2\,\min\{\delta(x,y),\delta(y,x)\}\leq (\gamma+1)\,\dH(x,y)=(\gamma+1)\,\dH(C),
	\end{equation*}
	which concludes the proof.
\end{proof}

\begin{remark}
The inequalities in the last part of Proposition~\ref{three} are typically not equalities. Take for example $p=0.1$ and $q=0.4$, from which $\gamma\approx 2.21$.
Let $C=\{(1,1,0,0), (1,0,1,1)\} \subseteq \FF^4$. We have $\hat\delta(C) \approx 1.79$, $\dH(C)=3$,
$\delta(C)\approx 4.21$ and $(\gamma+1)\dH(C)/2 \approx 4.81$. Therefore all the inequalities 
in Proposition~\ref{three} are strict for this code.
\end{remark}

The next result shows how bounds for codes with the Hamming metric translate into bound for codes endowed with discrepancy functions. In order to state the bounds in a general and compact form,
for $\Delta \in \R$ we let
\begin{align*}
    \mA(n,\Delta) &:=\max \{|C| \, : \, C \subseteq \F, \, |C| \ge 2, \, \delta(C) \ge \Delta\}, \\
    \hat\mA(n,\Delta) &:=\max \{|C| \, : \, C \subseteq \F, \, |C| \ge 2, \, \hat\delta(C) \ge \Delta\}, \\
    \mA^{\HH}(n,\Delta) &:=\max \{|C| \, : \, C \subseteq \F, \, |C| \ge 2, \, \dH(C) \ge \Delta\},
\end{align*}
where the maximum of the empty set is taken to be 1 (indicating a zero rate).
The following is an easy consequence of the last part of Proposition~\ref{three}.

\begin{proposition} \label{inh}
	For all $\Delta \in \R$ we have $$ \hat\mA(n,\Delta)\leq\mA^\HH(n,\Delta) \leq \mA(n,\Delta)\leq  \mA^\HH(n,2\Delta/(\gamma+1)).$$
\end{proposition}

\begin{remark}
Again, the inequalities in Proposition \ref{inh} are typically not equalities. For example, let $p=0.1$ and $q=0.4$, from which $\gamma\approx 2.21$. For $n=4$ and $\Delta=3$, we have $\hat\mA(4,3)=1$, $\mA^\HH(4,3)=2$, $\mA(4,3)=6$ and $\mA(4,6/(\gamma+1))=8$. Therefore all
the inequalities in Proposition~\ref{inh} are strict in this case.
\end{remark}

By combining the previous result with bounds from classical coding theory we obtain the following results.

\begin{corollary}
    Let $C\subseteq\F$ be a code. The following hold:
    \begin{enumerate}[label={(\arabic*)}]
    \setlength\itemsep{0.5em}
        \item $\displaystyle |C|\leq 2^{n-\left\lceil\frac{2\delta(C)}{\gamma+1}\right\rceil+1}$,
        \item $\displaystyle |C|\leq \frac{2^n}{\sum_{i=0}^t\binom{n}{i}}$, where $t$ is the largest integer with $t<\delta(C)/(\gamma+1)$,
        \item $\displaystyle  |C|\leq  \left\lfloor\frac{2d}{2d-n}\right\rfloor$, where $d= \lceil 2\delta(C)/(\gamma+1) \rceil$ and under the assumption that $2d>n$.
    \end{enumerate}
\end{corollary}
\begin{proof}
    The three bounds follow from the fact that $2\delta(C)/(\gamma+1)\leq \dH(C)$ by Proposition~\ref{three} and from the Singleton, the Hamming, and the Plotkin bound (respectively). We refer the reader to~\cite{macwilliams1977theory} for the statements of these bounds.
\end{proof}

It is interesting to observe that there are two classes of codes  for which the parameter $\gamma$ does not play a major role when focusing on bounds. These are linear codes and constant-weight codes (i.e., codes whose codewords all have the same Hamming weight), for which the theory fully reduces to Hamming-metric codes.

\begin{proposition}
    Let $C\subseteq \F$ be a code. The following hold:
    \begin{enumerate}[label={(\arabic*)}]
        \item If $C$ is linear, then $\delta(C)=\dH(C)$,
        \item If $C$ is constant-weight, then $\delta(x,y)=\frac{(\gamma+1)}{2}\,\dH(x,y)$ for all $x,y\in C$. In particular, $\delta(C)=\frac{(\gamma+1)}{2}\dH(C)$.
    \end{enumerate}
\end{proposition}
\begin{proof}
     Suppose that $C$ is linear. By Proposition \ref{three} we have $\delta(C)\geq \dH(C)$. Thus it remains to show that $\delta(C)\leq \dH(C)$. Since $C$ is linear, there exists $x\in C$ with $\dH(C)=\wH(x)$. We then have
\begin{equation*}
	\dH(C)=\wH(x)=\dH(0,x)=d_{01}(0,x)=\delta(0,x),
\end{equation*}
    which implies $\delta(C)\leq \dH(C)$, as desired. 
        
    Now suppose that $C$ is constant-weight and let $x,y \in C$. One can easily check that
    \begin{equation*}
    	d_{01}(x,y)=\wH(y)-d_{11}(x,y)=\wH(x)-d_{11}(x,y)=d_{10}(x,y).
    \end{equation*}
    Therefore we have
    \begin{align*}
        \delta(x,y)&=\gamma d_{10}(x,y)+d_{01}(x,y)=(\gamma+1)\,d_{01}(x,y),\\
        \dH(x,y)&=d_{10}(x,y)+d_{01}(x,y)=2\,d_{01}(x,y),
    \end{align*}
    which imply $\delta(x,y)=\frac{(\gamma+1)}{2}\,\dH(x,y)$. Finally, we have
    \begin{align*}
        \delta(C) &=\min\{\delta(x,y)\mid x,y\in C, \, x\neq y\} \\ &=\min\left\{\frac{(\gamma+1)}{2}\,\dH(x,y)\mid x,y\in C, \, x\neq y\right\} \\ &= \frac{(\gamma+1)}{2}\,\dH(C),
    \end{align*}
    concluding the proof.
\end{proof}

In the remainder of the section we use two arguments to relate the cardinality of a code to its minimum (symmetric) discrepancy. Our bounds involve an invariant that is finer than the cardinality, namely the weight distribution; see Definition \ref{def:weightdist}. 

\begin{lemma}
\label{lem:S1S2}
    Let $C\subseteq\F$ be a code. For all $x \in C$ define
    \begin{align*}
		S_1(x)&:=\left\{y\in\F \, : \, \delta(y,x)<\frac{\delta(C)+(\gamma-1)(\wH(y)-\wH(x))}{2}\right\},\\
		S_2(x)&:=\left\{y\in\F \, : \, \delta(y,x)<\frac{\hat\delta(C)+\wH(y)(\gamma-1)}{2}\right\}.
	\end{align*}
    Then for all $x,x'\in C$ with $x\neq x'$ we have
    $S_1(x)\cap S_1(x')=\emptyset$ and 
    $S_2(x)\cap S_2(x')=\emptyset$.
\end{lemma}
\begin{proof}
    Let $x,x'\in C$, with $x\neq x'$, and suppose towards a contradiction that there exists $y\in S_1(x)\cap S_1(x')$. Proposition \ref{prop:suff}\ref{item1:suff} implies that $x=D_C(y)=x'$, a contradiction. The proof that $S_2(x)\cap S_2(x')=\emptyset$ is analogous using Proposition~\ref{prop:suff}\ref{item2:suff}. 
\end{proof}

The next result provides two bounds relating the weight distribution of a code with its parameters $\delta(C)$ and $\hat\delta(C)$.

\begin{theorem}
\label{thm:hambound}
    Let $C\subseteq\F$ be a code. The following holds for all $i \in \{0, \ldots, n\}$:
	\begin{enumerate}[label={(\arabic*)}]
		\item \label{itm:hamone} $\displaystyle \sum_{j=0}^n \sum_{s\in \mathcal{S}\l\frac{\delta(C)+(\gamma-1)(i-j)}{2}\r}\lambda(i,j,s) \, \WH_j(C)\leq \binom{n}{i}$,
		\item $\displaystyle \sum_{j=0}^n \sum_{s\in \mathcal{S}\l\frac{\hat\delta(C)+(\gamma-1)i}{2}\r}\lambda(i,j,s) \, \WH_j(C)\leq \binom{n}{i}$.
	\end{enumerate}
\end{theorem}

\begin{proof}
    For $x\in C$ and $i \in \{0,\ldots,n\}$ define the sets
	\begin{align*}
		S_1(x,i)&:=\left\{y\in\F \, : \, \wH(y)=i \;\textup{ and }\;\delta(y,x)<\frac{\delta(C)+(\gamma-1)(i-\wH(x))}{2}\right\},\\
		S_2(x,i)&:=\left\{y\in\F \, : \, \wH(y)=i \;\textup{ and }\;\delta(y,x)<\frac{\hat\delta(C)+i(\gamma-1)}{2}\right\}.
	\end{align*}
	Lemma \ref{lem:S1S2} implies that $S_t(x,i)\cap S_t(x',i)=\emptyset$ for any $t\in\{1,2\}$ and all $x,x'\in C$ with $x\neq x'$. By Lemma \ref{lem:isom}, for all $x\in C$ of Hamming weight $j$ there exist 
	\begin{equation*}
		\lambda(i,j,s)=\Bin{j}{\frac{i\gamma -s+j}{\gamma+1}}\Bin{n-j}{\frac{s-j+i}{\gamma+1}}
	\end{equation*}
	vectors of $\F$ of Hamming weight $i$ and such that $\delta(y,x)=s$. This, along with the fact that the number of vectors of Hamming weight $i$ in $\F$ is $\binom{n}{i}$, implies
	\begin{align*}
	   \binom{n}{i}&\geq \left|\bigcup_{x\in C}S_1(x,i)\right|=\sum_{x\in C}|S_1(x,i)|=\sum_{j=0}^n \WH_j(C)\sum_{s\in \mathcal{S}\l\frac{\delta(C)+(\gamma-1)(i-j)}{2}\r}\lambda(i,j,s),\\
	   \binom{n}{i}&\geq \left|\bigcup_{x\in C}S_2(x,i)\right|=\sum_{x\in C}|S_2(x,i)|=\sum_{j=0}^n \WH_j(C)\sum_{s\in \mathcal{S}\l\frac{\delta(C)+(\gamma-1)i}{2}\r}\lambda(i,j,s),
	\end{align*}
	as desired.
\end{proof}

Theorem~\ref{thm:hambound} gives upper bounds on the cardinality of a code via integer linear programming. We illustrate this with an example.

\begin{example}
Let $p=0.1$ and $q=0.25$ which imply $\gamma\approx 1.57$. We wish to find the maximum cardinality of a code $C\subseteq\mathbb{F}_2^5$ with $\delta(C)=\gamma+2$. To this end, we solve an integer linear program in order to maximize $|C|$ under the constraints given by Theorem \ref{thm:hambound}\ref{itm:hamone}, for $0\leq i \leq 5$. For ease of notation, we write $\WH_j$ instead of $\WH_j(C)$ for $0\leq j\leq 5$.
We want to maximize the sum
    \begin{equation*}
        \Sigma= \WH_0+\WH_1+\WH_2+\WH_3+\WH_4+\WH_5
    \end{equation*}
    under the seven constraints:
    \begin{align*}
        \WH_0,\WH_1,\WH_2,\WH_3,\WH_4,\WH_5&\geq 0,\\
        \WH_0 + \WH_1 &\leq 1,\\
        5\WH_0 + \WH_1 + 2\WH_2 &\leq 5,\\
        4\WH_1 + \WH_2 + 3\WH_3 &\leq 10,\\
        3\WH_2 + \WH_3 + 4\WH_4 &\leq 10,\\
        2\WH_3 + \WH_4 + 5\WH_5 &\leq 5,\\
        \WH_4 + \WH_5 &\leq 1.
    \end{align*}
     One can check that the maximum value for $\Sigma$ is $4$, achieved for example by 
    \begin{align*}
        \WH_0=\WH_4&=1,\\
        \WH_1=\WH_2=\WH_5&=0,\\
        \WH_3&=2.
    \end{align*}
    The bound is in fact sharp and
    $C= \{(1,0,1,0,1),(0,0,0,1,1),(0,1,1,0,0),(1,1,0,1,0)\}$ is a code that meets it with equality.
\end{example}

Our next result is a constraint on the parameters of a code with given symmetric discrepancy. The proof is inspired by the classical Plotkin bound.

\begin{theorem}
\label{thm:lowPlotkin}
    Let $C\subseteq\F$ be a code. We have
    \begin{equation*}
        |C|\geq \left\lceil\frac{(\gamma+1)T(C)-\hat\delta(C)-n(\gamma-1)}{2n-\hat\delta(C)}\right\rceil,
    \end{equation*} 
    where $T(C)=\sum_{j=1}^nj\WH_j(C)$.
\end{theorem}

\begin{remark}
\label{rem:plotkin}
    The bound in the previous theorem is trivial when $\gamma=1$ and $C$ is a linear nondegenerate code. Indeed, if $C$ is such a code, then  $\sum_{j=1}^nj\WH_j(C)=n2^{k-1}$,
    where $k$ is the dimension of $C$. Moreover, $\hat\delta(C)=\dH(C)$ and since $\dH(C)\leq n$ we have 
    \begin{equation*}
        \left\lceil\frac{n2^k-\dH(C)}{2n-\dH(C)}\right\rceil\leq \left\lceil\frac{n2^k-\dH(C)}{n}\right\rceil=\left\lceil2^k-\frac{\dH(C)}{n}\right\rceil\leq 2^k= |C|.
    \end{equation*}
\end{remark}

\begin{proof}[Proof of Theorem \ref{thm:lowPlotkin}]
    We evaluate in two different ways the sum
    \begin{equation*}
        \Sigma:=\sum_{x\in C}\sum_{y\in C}\delta(y,x).
    \end{equation*}
    For ease of notation we denote by $M$ the cardinality of $C$ and write $\hat\delta$, $T$ and $\WH_j$ instead of~$\hat\delta(C)$, $T(C)$ and $\WH_j(C)$ respectively. We have
    \begin{align}
    \label{uu1}
        \Sigma &= \sum_{x\in C}\sum_{\tiny\begin{matrix}y\in C\\y\neq x\end{matrix}}\l\delta(y,x)-\wH(y)(\gamma-1)\r+\sum_{x\in C}\sum_{\tiny\begin{matrix}y\in C\\y\neq x\end{matrix}}\wH(y)(\gamma-1) \nonumber \\
        &\geq \sum_{x\in C}\sum_{\tiny\begin{matrix}y\in C\\y\neq x\end{matrix}}\hat\delta+\sum_{x\in C}\sum_{\tiny\begin{matrix}y\in C\\y\neq x\end{matrix}}\wH(y)(\gamma-1) \nonumber \\
        &= \hat\delta\, M(M-1)+(\gamma-1)(M-1)T \nonumber\\
        &\geq \hat\delta\,\frac{T}{n}(M-1)+(\gamma-1)(M-1)T, 
    \end{align}
    where the latter inequality follows from the fact that $nM\geq T$.
    On the other hand we have
    \begin{align}
     \Sigma &= \sum_{i=1}^n\sum_{x\in C}\sum_{y\in C}\delta(y_i,x_i) \nonumber \\
     &= \sum_{i=1}^n\sum_{\substack{x\in C\\x_i=1}}\sum_{\substack{y\in C\\y_i=0}}\delta(0,1)+\sum_{i=1}^n\sum_{\substack{x\in C\\x_i=0}}\sum_{\substack{y\in C\\y_i=1}}\delta(1,0) \nonumber \\
     &= \sum_{i=1}^n\sum_{\substack{x\in C\\x_i=1}}\sum_{\substack{y\in C\\y_i=0}}1+\sum_{i=1}^n\sum_{\substack{x\in C\\x_i=0}}\sum_{\substack{y\in C\\y_i=1}}\gamma. \label{last}
     \end{align}
Write $D_i(C):=|\{x\in C \, : \,  x_i=1\}|$
for $i \in \{1,\ldots,n\}$.
Then \eqref{last} becomes
$$\Sigma = (1+\gamma)\sum_{i=1}^n(M-D_i(C))D_i(C) = (1+\gamma)\l M\sum_{i=1}^n D_i(C)-\sum_{i=1}^n D_i(C)^2\r.$$
Observe moreover that
  \begin{equation*}        \sum_{i=1}^nD_i(C)=|\{(x,i) \, : \, x\in C, \, 1\leq i\leq n,\, x_i=1\}|=\sum_{x\in C}|\{i \, : \, 1 \leq i\leq n, \, x_i=1\}|=T.
    \end{equation*}
We now apply the Cauchy-Schwarz inequality to the vectors $(D_1(C),\ldots,D_n(C))$ and $(1,\ldots,1)$, obtaining
    \begin{equation*}
        n\l\sum_{i=1}^n D_i(C)^2\r\geq T^2.
    \end{equation*}
We thus have
    \begin{align}
    \label{eq:Plotup}
     \Sigma \leq(1+\gamma)T\l M-\frac{T}{n}\r.
    \end{align}
    The desired statement now follows by combining~\eqref{uu1} with~\eqref{eq:Plotup}, after tedious computations.
\end{proof}

We provide an example of a code 
that meets the bound of 
Theorem~\ref{thm:lowPlotkin}
with equality.

\begin{example}
    Let $p=0.1$, $q=0.4$, for which we have $\gamma\approx 2.21$. Consider the code $C=\{(1,1,1,1),(1,0,1,1),(1,0,1,0)\}$. We have $T(C)=9$ and $\hat\delta(C)=-2.63$.
    Therefore 
    \begin{equation*}
        \left\lceil\frac{(\gamma+1)T(C)-\hat\delta(C)-n(\gamma-1)}{2n-\hat\delta(C)}\right\rceil= \left\lceil\frac{3.21\cdot 9+2.63-4\cdot 1.21}{8+2.63}\right\rceil=\left\lceil 2.51\right\rceil=3=|C|.
    \end{equation*}
    In particular $C$  meets the bound of Theorem~\ref{thm:lowPlotkin} with equality.
\end{example}

\bigskip
\section{Conclusions}
In this paper, we used two notions of discrepancy between binary vectors to define new parameters of codes for the binary asymmetric channel. We then showed how these parameters measure the probability that the maximum likelihood decoder fails, and related them with more classical code parameters (such as the minimum Hamming distance). Finally, we derived bounds for the size of a code in terms of these parameters, giving examples of codes meeting the bounds with equality.

A natural open question is how to construct families of codes for the binary asymmetric channel having large cardinality and large minimum discrepancy simultaneously.

\bigskip
\section*{Acknowledgement}
The authors are grateful to the Referees of this paper for their very helpful suggestions and remarks.

\bigskip

\bigskip

\bibliographystyle{amsplain} 
\bibliography{BACbiblio.bib}
\end{document}